\documentclass[a4paper,oribibl]{llncs} 
\usepackage{fullpage}

\usepackage[english]{babel}
\usepackage{graphicx}
\usepackage{amsmath}
\usepackage{amssymb}
\usepackage{proof}

\usepackage{hyperref}

\usepackage{bbm}
\newcommand{\N}{{\mathbbm{N}}}

\newcommand{\C}{{\mathbbm{C}}}

\newcommand{\ket}[1]{{|#1\rangle}}
\newcommand{\bra}[1]{{\langle #1 |}}
\newcommand{\braket}[2]{{\langle #1 | #2 \rangle}}

\newcommand{\ketbraUm}[1]{{\ket{#1}\bra{#1}}}

\newcommand{\bolda}{{\mathbf{a}}}
\newcommand{\boldb}{{\mathbf{b}}}
\newcommand{\bolds}{{\mathbf{s}}}
\newcommand{\boldt}{{\mathbf{t}}}

\makeindex 
\begin{document}

\title{Embezzlement States are Universal for Non-Local Strategies}
\author{ \bf Mateus de Oliveira Oliveira } 
\institute{Blavatnik School of Computer Science, Tel Aviv University \\ mateusde@post.tau.ac.il} 
\maketitle

\begin{abstract} We prove that the family of embezzlement states defined by van
Dam and Hayden\cite{vanDamHayden2002} is universal for both quantum and classical entangled two-prover non-local 
games with an arbitrary number of rounds. More precisely, we show that for each 
$\varepsilon>0$ and each strategy for a $k$-round two-prover non-local game which 
uses a bipartite shared state on $2m$ qubits and makes the provers win with probability $\omega$, there exists 
a strategy for the same game which uses an embezzlement state on
$2m + 2m/\varepsilon$ qubits and makes the provers  win with probability
$\omega-\sqrt{2\varepsilon}$. Since the value of a game can be defined as the limit of 
the value of a maximal $2m$-qubit strategy as $m$ goes to infinity, 
our result implies that the classes $QMIP^*_{c,s}[2,k]$ and $MIP^*_{c,s}[2,k]$ 
remain invariant if we allow the provers to share only embezzlement states, for any completeness value
$c\in [0,1]$ and any soundness value $s < c$.  Finally we notice that 
the circuits applied by each prover may be put into a very simple universal form.
\end{abstract}

\section{Introduction}

A $k$-round non-local game is an interactive procedure involving a referee, and two 
provers Alice and Bob. At each round the referee randomly selects two questions 
drawn from finite sets and sends one of them to each prover. Subsequently, each of the 
provers replies to her/his question.  Alice and Bob are assumed to be in distinct locations
and not to be able to communicate. In this way, none of them knows which 
question was sent to the other.  At the end of the last round, 
the referee evaluates a publicly known predicate which depends on 
the whole history of questions and answers. 
The provers win if the predicate evaluates to true. The value of a non-local game is 
defined to be the maximum winning probability of Alice and Bob. 

The importance of non-local games is twofold. On one hand they are intimately connected
with multi-prover interactive proof systems 
\cite{Ben-OrGoldwasserKilianWigderson1988}. In these systems a polynomial time verifier
must decide the membership of a string $x$ in a language $L$ through an interactive protocol 
involving several provers which are not allowed to communicate. We say that a language 
$L$ has a $k$-round two-prover interactive proof system if there exists a polynomial time 
function which assigns to each $x$ a $k$-round game $G_x$ in such a way that if $x$ is in $L$ then 
the value of the game $G_x$ is above a threshold $c$, while if 
$x$ is not in $L$, the value of the game $G_x$ is below a threshold $s$ for $s<c$. 
We refer to $c$ as being the completeness of the system, and to $s$ as being 
its soundness. 

On the other hand, by allowing the provers to share a quantum system prepared in an 
arbitrary entangled state, non-local games become a suitable formalism to describe 
experiments that unveil the inherent non-locality of quantum mechanics. 
Following Bell's \cite{Bell1964} observation that some predictions of quantum mechanics 
are inconsistent with local hidden variables theories, several experiments were proposed 
with the aim to provide a decisive test between quantum mechanics and hidden local variables 
theories. As an example, in the CHSH game which is based on a thought experiment of 
Clauser, Horne, Shimony and Holt \cite{ClauserHorneShimonyHolt1969}, 
Bell's work implies that if the provers are classical the value of the game is $0.75$ while if
we allow the provers to share entanglement, there is a strategy which achieves a 
value of $\approx 0.85$. 
In other examples of games, like the Kochen-Specker game \cite{KochenSpecker1967,Strauman2008} and the Mermin-Peres magic square 
game \cite{Mermin1990,Aravind2004,Peres},
any classical strategy is doomed to fail with some probability while
there is a quantum strategy which always allows the provers to win.

When dealing with interactive proof systems it is customary to impose limits on the computational power of the verifier, 
while the provers are assumed to be at most limited by the laws of physics. In this sense, it is reasonable to consider 
interactive proof systems in which the provers are allowed to share arbitrary quantum states.
The study of how entanglement may affect the decidability properties of two-prover interactive 
proof systems was initiated by Cleve, Hoyer, Toner and Watrous \cite{CleveHoyerTonerWatrous2004}. 
They provide several examples of proof systems whose soundness is violated if we allow the provers 
to share an entangled state. Furthermore they provide evidences that 
entanglement may significantly interfere in the decidability properties of multi-prover interactive
proof systems. 
Let $\oplus MIP_{c,s}[2,1]$ denote the class of languages which can be decided by two prover interactive
proof systems in which the final decision of the verifier is taken solely based on the XOR of the 
$1$-bit answers of the provers, and $\oplus MIP_{c,s}^*[2,1]$ be its entangled version. 
Cleve et al. \cite{CleveHoyerTonerWatrous2004} show that $\oplus MIP^*_{c,s}[2,1] \subseteq EXP$, while in the classical case, it follows from works of H{\aa}stad \cite{Hastad2001}
and Bellare, Goldreich and Sudan \cite{BellareGoldreichSudan1998} that $\oplus MIP_{c,s}[2,1]=NEXP$ for certain completeness
and soundness values. Indeed, by 
combining a result of Wehner \cite{Wehner2006} and Jain, Upadhyay and Watrous \cite{JainJiUpadhyayWatrous2009}, 
it is possible to refine the first inclusion to $\oplus MIP^*[2,1] \subseteq PSPACE$. Thus unless, 
$PSPACE=NEXP$ entanglement indeed can weaken the decidability properties of XOR games. Entangled non local games 
were generalized and studied as well in the scenario in which the verifier is allowed to 
be quantum \cite{KempeKobayashiMatsumotoVidick2009, KempeKobayashiMatsumotoTonerVidick2008,KempeRegevToner2010}. 
In particular, some positive aspects of entanglement are explored in \cite{KempeKobayashiMatsumotoVidick2009}, where 
the authors provide some evidence that prior entanglement may be useful for honest provers.  

In order to make 
the study of entangled games slightly easier, it is reasonable to ask whether the bipartite state shared by the provers
may be restricted to a class of states which is easy to describe and to work with. The aim of this work 
is to show that the embezzlement family of states defined by van Dam and Hayden\cite{vanDamHayden2002}, satisfy these criteria. 
More precisely, in Theorem \ref{theorem:SpecialState} we prove that the family of embezzlement states is universal for two-prover non-local 
games with any  number of rounds, in the sense that any strategy for a two-prover non-local game which yields a 
value $\omega$ may be replaced by a strategy for the same game that uses an embezzlement state and that 
yields a value of at least $\omega - \sqrt{2\varepsilon}$ for any $\varepsilon$ with 
$0 < \varepsilon < 1$ with only a linear, in $1/\varepsilon$, overhead  on the number of qubits to be shared. 
Since the value $\omega$ of a game can be defined as the limit of the value of a maximal $2m$-qubit strategy as 
$m$ goes to infinity, this implies that $\omega$ itself is not changed when only embezzlement states are considered. 
As a consequence, the classes $QMIP^*_{c,s}[2,k]$ and $MIP^*_{c,s}[2]$ remain invariant 
through our restriction (Corollary \ref{corollary:Classes}). Finally, as an observation, we note in 
Theorem \ref{theorem:UniversalStrategy} that the circuits applied by the provers may also be put into a very simple 
universal form. 

While in the classical case a series of results \cite{BabaiLundFortnow1991,CaiCondonLipton1990,LapidotShamir1991,Feige1991,FeigeLovasz1992} 
established the relation  $MIP[2,k]=MIP[2,1]=NEXP$ for any $k$, in the setting 
in which the provers share entanglement it makes sense to consider interactive proof systems with 
an arbitrary number of rounds because in this case it is not known whether $MIP^*[2,k]=MIP^*[2,1]$ for
$k\geq 2$. It is also worth noting that the embezzlement family has been already considered (and generalized 
to any constant number of provers) by Leung, Toner and Watrous \cite{LeungTonerWatrous2008} 
and used to prove that if we allow the referee to be quantum, then there are 
one-round games whose value cannot be achieved by means of strategies that share a finite amount of entanglement.
Nevertheless, the embezzlement family seems to have passed unnoticed as a universal family of states for 
non-local games, and in some of the literature concerning entangled multiprover interactive proof systems,
the existence of such family is implicitly stated as an open problem \cite{KobayashiMatsumoto2003}.

The rest of this paper is organized as follows: In Section \ref{section:NonLocalGames} we provide a formal definition of non-local games. In 
Section \ref{section:Embezzlement} we introduce van Dam and Hayden's embezzlement family and 
prove our universality results (Theorems \ref{theorem:SpecialState} and \ref{theorem:UniversalStrategy}, and Corollary \ref{corollary:Classes}).

\section{Non-Local Games} \label{section:NonLocalGames}

A $k$-round two-prover non-local game is an interactive procedure undertaken by a verifier
and two provers which we call Alice and Bob. The game proceeds as follows: Given
two sets of questions $S$ and $T$, two sets of answers $A$ and $B$, and a
predicate $V\subseteq S^k\times T^k\times A^k\times B^k$, at round $i$ the verifier choses a pair of
questions $(s_i,t_i)\in S\times T$ accordingly to a probability distribution $\pi_i$ 
defined on $S\times T$ and sends the question $s_i$ to Alice and the question $t_i$ to Bob. Alice replies with an answer $a_i \in A$ 
and Bob replies with an answer $b_i \in B$. The provers win the game if the history 
$(s_1...s_k,t_1...t_k,a_1...a_k,b_1...b_k)$ of all questions and answers satisfies the predicate $V$. 
The goal of the provers is to follow a strategy that maximizes their winning probability. 
We note that the probability distribution $\pi_i$ with which the verifier choses the questions at
round $i$ may depend on the questions and answers from previous rounds. We denote a $k$-round 
non-local game by $G=(V,\pi)$ where $\pi$ is a set of probability distributions over $S\times T$

\begin{equation}
\label{equation:ProbabilityDistribution}
\pi = \{ \pi_i(s_1...s_i,t_1...t_i,a_1...a_i,b_1...b_i) | 1\leq i \leq k-1\}
\end{equation}

The provers' strategies can be described by Positive Operator Valued Measurements (POVM's). 
Formally, a POVM in $\C^n$ with outcomes in $\mathcal{I}$ is a family of $n$-dimensional operators $M=\{M_i\}_{i\in \mathcal{I}}$ 
satisfying $\sum M_i^{\dagger}M_i=I_{n}$, where $I_n$ is the identity in $\C^n$. Measuring a quantum system 
prepared in a state $\ket{\psi}\in \C^{n}$ according to $M$, yields the outcome $i$ with probability 
$\bra{\psi}M_i^{\dagger}M_i\ket{\psi}$ and post-measurement state $M_i\ket{\psi}/ \bra{\psi}M_i^{\dagger}M_i\ket{\psi} $
  \cite{NielsenChuang2000}.

In a quantum strategy, the provers share a quantum register consisting of $2m$ qubits
prepared in a bipartite state $\ket{\psi}\in \C^{2^m} \otimes \C^{2^m}$ in such a way that each prover holds $m$ of 
the qubits. The state shared by the provers can be assumed to be pure, for if it were mixed, we could simply consider a 
pure state in a higher dimensional Hilbert space. For each question $s_i\in S$ and history of questions and answers $(s_1...s_{i-1}, a_1...a_{i-1})$, 
Alice has a POVM $\{ X_{s_1...s_{i-1}|s_i}^{a_1...a_{i-1}|a_i}\}_{a_i\in A}$ with outcomes in $A$. 
Similarly, for each question $t_i\in T$ and history of questions and answers $(t_1...t_{i-1}, b_1...b_{i-1})$, 
Bob has a POVM $\{ Y_{t_1...t_{i-1}|t_i}^{b_1...b_{i-1}|b_i}\}_{b_i\in B}$ with outcomes in $B$. In a slight abuse of 
notation we simply write $\{X_{s_i}^{a_i}\}_{a_i\in A}$ and $\{Y_{t_i}^{b_i}\}_{b_i\in B}$ whenever the history of the previous rounds is clear. 
A strategy on $2m$-qubits is completely determined by a triple $(\ket{\psi}_{2m},X,Y)$ where $\ket{\psi}$ is the 
shared state, $X$ is the collection of all POVM's of Alice and $Y$ of all POVM's of Bob. 
Let $\ket{\psi}=\ket{\psi_1}$ be the initial state shared by the provers and $\ket{\psi_i}$ be the state shared by the provers at the $i$-th round. The probability with which Alice and Bob reply respectively $a_i$ and $b_i$ at the $i$-th round when questioned 
with $s_i$ and $t_i$ is given by

\begin{equation}
\label{equation:probabilityAnswer}
\bra{\psi_{i}} (X_{s_i}^{a_i}\otimes Y_{t_i}^{b_i})^{\dagger} X_{s_i}^{a_i}\otimes Y_{t_i}^{b_i}\ket{\psi_i}
\end{equation} 

and the new state becomes 

\begin{equation}
\label{equation:recurrenceState}
\ket{\psi_{i+1}} =\frac{ X_{s_i}^{a_i} \otimes Y_{t_i}^{b_i}\ket{\psi_{i}}}{ 
\bra{\psi_i} (X_{s_i}^{a_i}\otimes Y_{t_i}^{b_i})^{\dagger} X_{s_i}^{a_i}\otimes Y_{t_i}^{b_i} \ket{\psi_i}}\; . 
\end{equation} 

As a convention we let boldface letters range over $k$-tuples of elements: $\bolds\in S^k$, $\boldt\in T^k$, $\bolda\in A^k$ and $\boldb\in B^k$.
The value of the strategy $(\ket{\psi},X,Y)$ for the game is defined as

\begin{equation}
\label{equation:ValueGame}
\omega_G(\ket{\psi},X,Y)= \sum_{\bolds,\boldt,\bolda,\boldb} V(\bolds,\boldt,\bolda,\boldb) \prod_{i=1}^k\pi_i(s_i, t_i) 
\prod_{i=1}^k\bra{\psi_{i}} (X_{s_i}^{a_i}\otimes Y_{t_i}^{b_i})^{\dagger} X_{s_i}^{a_i}\otimes Y_{t_i}^{b_i}\ket{\psi_i},
\end{equation}

which by using Equations (\ref{equation:probabilityAnswer}) and (\ref{equation:recurrenceState}), can be rewritten explicitly as a function of 
the initial shared state $\ket{\psi}$ as  

\begin{equation}
\label{equation:ValueGame}
\omega_G(\ket{\psi}_{2m}, X,Y)= \sum_{\bolds, \boldt, \bolda,\boldb } \left[V(\bolds,\boldt,\bolda,\boldb) \prod_{i=1}^k \pi_i(s_i, t_i) \right] 
\bra{\psi} \left( X_{s_1}^{a_1}...X_{s_k}^{a_k} \otimes Y_{t_1}^{b_1}...Y_{t_k}^{b_k} \right)^{\dagger}
 X_{s_1}^{a_1}...X_{s_k}^{a_k} \otimes Y_{t_1}^{b_1}...Y_{t_k}^{b_k} \ket{\psi}. 
\end{equation}

The {\em entangled value} of $G$ is defined as the limit of the maximum value among all $n$-qubit strategies as 
$n$ goes to infinity. 

\begin{equation} 
\label{equation:ValueEntangled} \omega_G^e=\lim_{m\rightarrow
\infty} \max_{\ket{\psi}_{2m},X,Y} \omega_G^e(\ket{\psi}_{2m},X,Y). \end{equation}

Non-local games can be further generalized to the case in which the verifier has quantum 
capabilities. In this case the communication with the provers proceeds through the
exchange of quantum registers. 

\begin{definition}[Quantum Entangled Non-Local Games]
\label{definition:QuantumNonLocalGames}
A $k$-round $2$-prover entangled quantum game $G(V_1,...,V_k)$ is defined by a verifier strategy
$(V_1,...,V_k,V_{k+1})$ where each $V_i$ is a quantum circuit acting on three quantum registers: 
A private quantum register $priv_V$ and two quantum communication registers $com_X$ and $com_Y$. 
A $2m$-qubit strategy for $G$ consists of a bipartite quantum state $\ket{\psi} \in \C^{2^m}\otimes C^{2^m}$, 
and two sequences of quantum circuits $X=(X_1,...,X_k)$ and $Y=(Y_1,...,Y_k)$ (prover's circuits), where each $X_i$ acts on the quantum communication
register $com_X$ and on a private quantum register $priv_X$, and each $Y_i$ is a quantum circuit acting on the communication 
register $com_Y$ and on a private quantum register $priv_Y$. 

The game proceeds as follows: At the start, $priv_V$, $com_X$ and $com_Y$
are initialized to $\ket{0}$ while $priv_X$ and $priv_Y$ are initialized to the bipartite state $\ket{\psi}$. 
The $j$-th round of the game consists in the application of the circuit $V_j$ followed by the application of $X_j$ and $Y_j$ to their 
respective registers. After the $k$-th round, $V_{k+1}$ is applied and the first private qubit $q$ of the verifier is measured in 
the computational basis. 
\end{definition}

The quantum value $\omega_q(\ket{\psi}_{2m},X,Y)$ of a strategy $(\ket{\psi}_{2m},X,Y)$ is defined as the probability with which 
the measured qubit $q$ is $\ket{1}$. Similarly to the classical entangled case, the value of a $k$-round quantum game 
$G=(V_1,...,V_k,V_{k+1})$ is defined as 

\begin{equation}
\label{equation:ValueQuantum}
\omega_G^q= \lim_{m\rightarrow \infty} \max_{\ket{\psi}_{2m},X,Y} \omega_G^q(\ket{\psi}_{2m},X,Y). 
\end{equation}
In the most general case, the circuits corresponding to both the verifier and the provers may contain any kind of 
physically realisable operations. However such circuits may be efficiently simulated by quantum 
circuits consisting only of unitary operations followed by a single measurement \cite{AharonovKitaevNisan1998}. 
Furthermore, by considering higher dimensional Hilbert spaces, we may assume that the state shared by the provers is pure. 

Classical entangled games may be cast as a subclass of quantum entangled games: Each verifier circuit
consists of a measurement of the communication registers $com_X$ and $com_Y$ in the computational basis, 
followed by the application of a permutation of the basis states. 
The formulation of classical entangled two-prover non-local games in terms of predicates is more natural, 
and allow us to define the value of the entangled game by a closed formula, which is completely circuit 
independent. Nevertheless the reformulation of classical 
entangled games as a special case of quantum entangled games is more suitable for the goals of this paper. 
In particular, the proof of Theorem \ref{theorem:SpecialState} turns out to be much simpler in this setting.

\begin{definition}[Quantum (Classical) Entangled Multiprover Interactive Proof Systems]
\label{definition:MultiproverInteractiveProofSystem}
A language $L$ over an alphabet $\Sigma$ can be decided by a $k$-round quantum (classical) entangled two-prover interactive proof system with completeness $c$ and 
soundness $s$ if there exists a deterministic polynomial time algorithm $P$ that on input $x\in \Sigma^*$ constructs the description 
of the circuits of a quantum (classical) entangled $k$-round two-prover non-local game $G=(V_1,...,V_k,V_{k+1})$, such that if $x\in L$ then $\omega_G^q \geq c$ 
($\omega_G^e \geq c$) and if $x\notin L$ then  $\omega_G^q \leq s$  ($\omega_G^e \leq s$). 
\end{definition}

We denote by $QMIP^{*}_{c,s}[2,k]$ and $MIP^{*}_{c,s}[2,k]$ the classes of all languages that have a quantum, resp. classical, entangled 
$k$-round two-prover interactive proof system with completeness $c$ and soundness $s$.

\section{Universality of the Family of Embezzlement States }
\label{section:Embezzlement}

Embezzlement states were defined in \cite{vanDamHayden2002} as follows: 

\begin{equation} 
\ket{\mu}_{2n} = \frac{1}{C}\sum_{j=1}^{2^n}
\frac{1}{\sqrt{j}}\ket{j}_n\ket{j}_n \hspace{1cm} C=\sqrt{\sum_{j=1}^{2^n}{\frac{1}{j}}}  . 
\end{equation}

Let  $\ket{\psi}_{2n} = \sum_{i=1}^{2^m} \alpha_i \ket{\theta_i}\ket{\theta_i}$ be a bipartite $2m$-qubit state written according
to its Schmidt decomposition. Then the state $\ket{\mu}\otimes \ket{\psi}$ admits a Schmidt decomposition of the form

\begin{equation}
\label{equation:TensorProduct}
\sum_{j,i} \gamma_{j,i} \ket{j}\ket{j}\ket{\theta_i}\ket{\theta_i}
\end{equation}

Let $\gamma_{j_1,i_1} \geq \gamma_{j_2,i_2} \geq ... \geq \gamma_{j_{N},i_{N}}$ be the $N=2^n$ largest coefficients of the above Schmidt decomposition. 
Then define the $n-th$ embezzled version of $\ket{\psi}$ to be the state 

\begin{equation}
\label{equation:IntermediaryState}
\ket{E(\psi)}_{2n,2m} = \sum_{r=1}^{2^n} \frac{1}{\sqrt{r}} \ket{j_r}\ket{j_r}\ket{\theta_{i_r}}\ket{\theta_{i_r}}.
\end{equation}

We note that Alice and Bob may transform the state  $\ket{\mu}_{2n}$ into the state $\ket{E(\psi)}_{2n,2m}$ by performing 
only local operations and without communication. First each prover prepares a local ancilla register of size $m$ in the state $\ket{1}_m$, so 
that $\ket{\mu}_{2n}$ becomes $\ket{\mu}_{2n}\otimes\ket{1}_m\ket{1}_m$. Subsequently both Alice and Bob apply a unitary $U$ that maps each basis 
state $\ket{j}_n\ket{1}_m$ to the basis state $\ket{j_r}_n\ket{\theta_{i_r}}_m$, thus transforming 
$\ket{\mu}_{2n}\otimes\ket{1}_m\ket{1}_m$ into $\ket{E(\psi)}_{2n,2m}$. Surprisingly, as stated in the next theorem, by increasing $n$ the state $\ket{E(\psi)}_{2n,2m}$ can 
be made arbitrarily close to $\ket{\mu}_{2n}\otimes\ket{\psi}_{2m}$. 

\begin{theorem}[Embezzlement \cite{vanDamHayden2002}]
\label{theorem:Embezzlement} Let $\ket{\psi}_{2m}=\sum_{j=1}^{2^{m}} \alpha_i \ket{\theta_i}\ket{\theta_i}$ 
be a $2m$ qubit bipartite state written according to its Schmidt decomposition, $\varepsilon$ be such that $0<\varepsilon < 1$;  
and $n,m\in \N$ be such that $n \geq \frac{m}{\varepsilon}$. 
Then $(\bra{\mu}_{2n} \otimes \bra{\psi}_{2m} ) \ket{ E(\psi)}_{2n,2m}  \geq 1-\varepsilon.$
\end{theorem}

To show our main theorem, we need some more notation: 
The trace distance between two states $\ket{\psi}$ and $\ket{\phi}$ in $\C^n$ is defined as 
$D(\ket{\psi},\ket{\phi})=\frac{1}{2}tr|(\ketbraUm{\psi}-\ketbraUm{\phi} )|$ where $|A|\equiv \sqrt{A^{\dagger}A}$. 
If $\{M_i\}_{i\in \mathcal{I}}$ is a POVM with outcomes in $\mathcal{I}$ and $p_i=\bra{\psi}M_i^{\dagger}M_i\ket{\psi}$ and $q_i=\bra{\phi} M_i^{\dagger}M_i \ket{\phi}$  
are the probability distributions induced by the measurement on $\ket{\psi}$ and $\ket{\phi}$ respectively, 
then $D(p_i,q_i)\leq D(\ket{\psi},\ket{\phi})$ where $D(p_i,q_i)=\frac{1}{2}\sum_i |p_i - q_i|$ is the classical total variance distance 
between the probability distributions $p_i$ and $q_i$ (see for example theorem $9.1$ of \cite{NielsenChuang2000} for a proof). 
In other words if two quantum states are close in trace distance, then any measurement performed on those states will give rise to probability
distributions which are close in the classical sense. Also it can be proved that 
$D(\ket{\psi}, \ket{\phi}) \leq \sqrt{1-\braket{\psi}{\phi}^2}$ and thus 
if $\braket{\psi}{\phi}\geq 1-\varepsilon$, then $D(\ket{\psi},\ket{\phi}) < \sqrt{2\varepsilon}$. 

Next we prove our main theorem. It says that the value of a quantum strategy for a quantum entangled non-local game in which 
the provers share a pure state $\ket{\psi}$ on $2m$ qubits can be arbitrarily approximated by the value of a strategy 
in which the provers share an embezzlement state. Since classical entangled games can be regarded as a special case 
of quantum entangled games, Theorem \ref{theorem:SpecialState} holds also in the classical entangled setting. 

\begin{theorem} 
\label{theorem:SpecialState} 
Let $(\ket{\psi}_{2m},X,Y)$ be a $2m$-qubit quantum strategy for a $k$-round two-prover non-local game $G(V_1, ... ,V_k, V_{k+1})$. 
Then for any $\varepsilon$ with $0 < \varepsilon < 1$ there exists a $2m(1+1/\varepsilon)$-qubit strategy 
$(\ket{\mu}_{2m/\varepsilon}\otimes\ket{1}_m\ket{1}_m ,X',Y')$ such that $\omega_G^q(\ket{\mu}_{2m/\varepsilon}\ket{1}_m\ket{1}_m, X', Y') \geq 
\omega_{G}^q(\ket{\psi}_{2m}, X, Y) - \sqrt{2\varepsilon}$.  
\end{theorem}
\begin{proof}
Let $(\ket{\mu}_{2m/\varepsilon}\otimes\ket{\psi}_{2m}, \overline{X}, \overline{Y})$ be a strategy for $G$ 
where $\ket{\mu}_{2m/\varepsilon}$ is the embezzlement state and $\overline{X}$ and $\overline{Y}$ are obtained 
by tensoring each circuit in $X$ and each circuit in $Y$ with the identity on $m/\varepsilon$ qubits acting on 
half of the qubits of $\ket{\mu}_{2m/\varepsilon}$. Then clearly 
$\omega_G^q(\ket{\mu}_{2m/\varepsilon}\otimes \ket{\psi}_{2m}, \overline{X},\overline{Y}) = \omega_G^q(\ket{\psi}_{2m},X,Y)$.
By Definition \ref{definition:QuantumNonLocalGames}, the interplay of the verifier's strategy with the provers's strategies,
prior to the final measurement of the verifier, may be regarded as the application of a single unitary $U_G$ to a pure state.
Let $\ket{E(\psi)}_{2m/\varepsilon,2m}$ be the embezzled version of $\ket{\psi}$ as defined in Equation (\ref{equation:IntermediaryState}) and 
set $\ket{\phi}=U_G \ket{E(\psi)}_{2m/\varepsilon,2m}$ and $\ket{\phi'}= U_G \ket{\mu}_{2m/\varepsilon}\otimes \ket{\psi}$. 
Since by Theorem \ref{theorem:Embezzlement}, $(\bra{\mu}_{2m/\varepsilon,2m} \otimes \bra{\psi}_{2m})\ket{E(\psi)}_{2m/\varepsilon,2m} \geq 1 -\varepsilon$, 
we have $\braket{\phi}{\phi'} \geq 1 -\varepsilon$ and the trace distance $D(\ket{\psi},\ket{\phi}) < \sqrt{2\varepsilon}$. 
Let $\{M_i\}_{i\in \mathcal{I}}$ be the POVM measurement made by the verifier in the end of the $k$-th round and let 
$p_i=\bra{\phi}M_i^{\dagger}M_i \ket{\phi}$ and $q_i=\bra{\phi'}M_i^{\dagger}M_i \ket{\phi'}$. Then 
$D(p_i,q_i)\leq D(\ket{\phi},\ket{\phi'}) \leq \sqrt{2\varepsilon}$. 
Finally there is a unitary $U$ such that $U\otimes U \ket{E(\psi)}_{2m/\varepsilon,2m}=\ket{\mu}_{2m/\varepsilon}\otimes \ket{1}_m\ket{1}_m$ 
where one of the $U$'s acts on Alice's qubits and the other on Bob's qubits. Then the final strategy is 
$(\ket{\mu}_{2m/\varepsilon,2m},X',Y')$ where $X'=U\overline{X}U^{\dagger}$ and $Y'=U\overline{Y}U^{\dagger}$.
$\square$
\end{proof}

As pointed out in the introduction, Leung, Toner and Watrous \cite{LeungTonerWatrous2008} showed that there are quantum entangled games whose value 
is never attained by a strategy whose shared state has a constant number of qubits, and thus the limit in Equation (\ref{equation:ValueQuantum}) is 
fundamental. It is still not known whether the same situation holds for classical entangled games. 
Despite the fact that Theorem \ref{theorem:SpecialState} concerns only strategies 
with a finite number of qubits, it is still possible to prove that the limit in Equations \ref{equation:ValueEntangled} and  \ref{equation:ValueQuantum}
does not change if we consider only embezzlement states. This in particular implies that the classes $QMIP^*_{c,s}[2,k]$  and $MIP^*_{c,s}[2,k]$ remain 
invariant if we allow the provers to share only embezzlement states. Let $QMIP^{E^*}_{c,s}[2,k]$ ($MIP^{E^*}[2,k]$) be the class 
of languages that can be decided by quantum (classical) entangled $k$-round two-prover interactive proof systems whose provers are only allowed to share embezzlement states. 

\begin{corollary}
\label{corollary:Classes}
For any completeness value $c\in[0,1]$ and any soundness value $s<c$,  
$QMIP^{E^*}_{c,s}[2,k]$ ($MIP^{E^*}_{c,s}[2,k]$) is equal to $QMIP^{*}_{c,s}[2,k]$ ($MIP^{*}_{c,s}[2,k]$). 
\end{corollary}
\begin{proof}
Let $L$ be a language in $QMIP^*_{c,s}[2,k]$ ($MIP^*_{c,s}[2,k]$). It is enough to prove that for any $x\in L$ the value of the game $G_x$ associated 
to $x$ remains the same if we restrict the state shared by the provers to belong to the embezzlement family. 
Since the proof holds both for classical entangled and for quantum entangled games, we write simple $\omega_{G_x}$ 
for the value of $G_x$.
If $\omega_{G_x}$
is reached by a strategy in which the provers share a finite dimensional state $\ket{\psi}_{2m}$, then by Theorem \ref{theorem:SpecialState}
there exist a sequence of strategies sharing states $\ket{\mu}_{2n}\otimes\ket{1}_m\ket{1}_m$ whose value approaches $\omega_{G_x}$  as 
$n\rightarrow \infty$. 
Now suppose that there is no finite dimensional strategy whose value is $\omega_{G_x}$, and let $\omega_2,\omega_4,...,\omega_{2m},...$ 
be an infinite sequence where $\omega_{2m}$ is the maximum value among strategies sharing a quantum states on $2m$ qubits. Then by Theorem 
\ref{theorem:SpecialState}, for any two such consecutive values $\omega_{2(m-1)}$ and $\omega_{2m}$, and for a small enough $\varepsilon$, 
there exists a strategy on $2(1+1/\varepsilon)m$ qubits whose value $\omega \geq \omega_{2m} - \sqrt{2\varepsilon}$ is between 
$\omega_{2(m-1)}$ and  $\omega_{2m}$. 
$\square$
\end{proof}

In Theorem \ref{theorem:UniversalStrategy} we state a dual of Theorem \ref{theorem:SpecialState} which says that the circuits applied by the provers 
can be put into a universal form.

\begin{theorem}[Universal Strategy]
\label{theorem:UniversalStrategy}
For each $k$ and each $\varepsilon>0$ there is a universal set of $k$-round circuits $\{(\mathcal{X}_M,\mathcal{Y}_M)\}_{M\in \N}$ such that 
for each $k$-prover non-local game $G$ and each strategy $(\ket{\psi}_{2m},X,Y)$, there is a $M\in \N$ and a state $\ket{\psi}_{2m}\ket{A}_M\ket{B}_M$, such that 
$$\omega_G(\ket{\psi}_{2m}\ket{A}_M\ket{B}_M,\mathcal{X}_M,\mathcal{Y}_M)\geq \omega_G(\ket{\psi}_{2m}, X,Y) -\varepsilon.$$
\end{theorem}
\begin{proof}
Any unitary matrix acting on $d$ qubits can be $\varepsilon$-approximated by a circuit with $poly(2^d,\log{1/\varepsilon})$ gates from the universal 
set of gates $\{CNOT, H, \pi/8\}$  \cite{BoykinMorPulverRoychowdhuryVatan1999}. By adding the $SWAP$ gate to this set, such circuits
can be put into a nearest neighbor configuration, in which the $CNOT$ and $SWAP$ operate only on adjacent pair of qubits.
Alice and Bob hold two registers each: one working register with $d=m+v$ qubits, where $v$ is the size of the communication register with the verifier, 
and an ancilla register of size $M=k\cdot poly(2^d,\log{k/\varepsilon})$ divided into $k$ regions with equal number of qubits. 
Each prover regards the state of the $j$-th region of her/his  ancilla register as a program which will determine the unitary that will be applied to her/his working register
at the $j$-th round. More precisely, each prover applies a circuit
of the form $C_MC_{M-1}...C_1$ where each $C_i$ is a controlled gate which applies one of the four gates $SWAP,CNOT,I\otimes H$ or $I\otimes \pi/8$ to 
qubits $2i (\!\!\!\mod d)$ and $2i +1 (\!\!\!\mod d)$ \footnote{For simplicity we assume that the first and last qubits of the working register are adjacent.} 
 of the working register depending whether the state of qubits $2i$ and $2i+1$ of the 
ancilla register is $\ket{00},\ket{01},\ket{10}$ or $\ket{11}$ respectively. For some configuration of $\ket{A},\ket{B}$ of 
the ancillae registers of Alice and Bob respectively, each unitary in $\mathcal{X}_M$ (resp. $\mathcal{Y}_M$) will be $\varepsilon/k$-close to its 
corresponding unitary in $X$ (resp. $Y$). Since there are $k$ rounds and the errors accumulate additively, 
$\omega_G(\ket{\psi}\ket{A}\ket{B},\mathcal{X}_{M},\mathcal{Y}_{M}) \geq \omega(\ket{\psi},X,Y)-\varepsilon$. 
 $\square$
\end{proof}

\section{Acknowledgements}

The author thanks Julia Kempe and Thomas Vidick for their valuable revisions on drafts of this paper, and Andr\'{e} Chailloux, Iordanis Kerenidis and Fr\'{e}d\'{e}rick Magniez, 
for useful discussions.  The author also acknowledges support by Julia Kempe's Israel Science Foundation grant and by Julia Kempe's European 
Research Council (ERC) Starting  Grant QUCO as well as by the Wolfson Family Charitable Trust.

\bibliographystyle{abbrv}
\bibliography{embezzlementUniversal}

\begin{thebibliography}{10}

\bibitem{AharonovKitaevNisan1998}
D.~Aharonov, A.~Kitaev, and N.~Nisan.
\newblock Quantum circuits with mixed states.
\newblock In {\em Proceedings of the 30nd ACM Symposium on Theory of Computing
  ({STOC})}, pages 20--30, 1998.

\bibitem{Aravind2004}
P.~K. Aravind.
\newblock Quantum mysteries revisited again.
\newblock {\em American Journal of Physics}, 72:1303--1307, 2004.

\bibitem{BabaiLundFortnow1991}
L.~Babai, L.~Fortnow, and C.~Lund.
\newblock Non-deterministic exponential time has two-prover interactive
  protocols.
\newblock {\em Computational Complexity}, 1:3--40, 1991.

\bibitem{Bell1964}
J.~Bell.
\newblock On the {E}instein-{P}odolsky-{R}osen paradox.
\newblock {\em Physics}, 1:3:195--200, 1964.

\bibitem{BellareGoldreichSudan1998}
Bellare, Goldreich, and Sudan.
\newblock Free bits, {PCP}s, and nonapproximability--towards tight results.
\newblock {\em SICOMP: SIAM Journal on Computing}, 27, 1998.

\bibitem{Ben-OrGoldwasserKilianWigderson1988}
M.~Ben-Or, S.~Goldwasser, J.~Kilian, and A.~Wigderson.
\newblock Multi-prover interactive proofs: How to remove intractability
  assumptions.
\newblock In {\em Proceedings of the 20th Annual ACM Symposium on Theory of
  Computing ({STOC})}, pages 113--131, 1988.

\bibitem{BoykinMorPulverRoychowdhuryVatan1999}
P.~O. Boykin, T.~Mor, M.~Pulver, V.~Roychowdhury, and F.~Vatan.
\newblock On universal and fault-tolerant quantum computing.
\newblock In {\em Proceedings of the 40th Annual IEEE Symposium on Foundations
  of Computer Science ({FOCS})}, pages 486--494. Society Press, 1999.

\bibitem{CaiCondonLipton1990}
J.-Y. Cai, A.~Condon, and R.~J. Lipton.
\newblock On bounded round multi-prover interactive proof systems.
\newblock In {\em Prooceedings of the Structure in Complexity Theory Conference
  ({C}o{C}o)}, pages 45--54, 1990.

\bibitem{ClauserHorneShimonyHolt1969}
J.~Clausser, M.~Horne, A.~Shimony, and R.~Holt.
\newblock Proposed experiment to test local hidden-variable theories.
\newblock {\em Physical Review Letters}, 23:880--884, 1969.

\bibitem{CleveHoyerTonerWatrous2004}
R.~Cleve, P.~H{\o}yer, B.~Toner, and J.~Watrous.
\newblock Consequences and limits of nonlocal strategies.
\newblock In {\em Proceedings of the 19th IEEE Conference on Computational
  Complexity ({CCC})}, pages 236--249. IEEE Computer Society, 2004.

\bibitem{Feige1991}
U.~Feige.
\newblock On the success probability of the two provers in one-round proof
  systems.
\newblock In {\em Proceedings of the Structure in Complexity Theory Conference
  ({C}o{C}o)}, pages 116--123, 1991.

\bibitem{FeigeLovasz1992}
U.~Feige and L.~Lov{\'a}sz.
\newblock Two-prover one-round proof systems: Their power and their problems.
\newblock In {\em Proceedings of the 24th Annual ACM Symposium on Theory of
  Computing ({STOC})}, pages 733--744, 1992.

\bibitem{Hastad2001}
J.~H{\aa}stad.
\newblock Some optimal inapproximability results.
\newblock {\em Journal of the ACM}, 48(4):798--859, 2001.

\bibitem{JainJiUpadhyayWatrous2009}
R.~Jain, Z.~Ji, S.~Upadhyay, and J.~Watrous.
\newblock $\mbox{QIP = PSPACE}$.
\newblock In {\em Proceedings of the 42nd ACM Symposium on Theory of Computing
  ({STOC})}, pages 573--582, 2010.

\bibitem{KempeKobayashiMatsumotoTonerVidick2008}
J.~Kempe, H.~Kobayashi, K.~Matsumoto, B.~Toner, and T.~Vidick.
\newblock Entangled games are hard to approximate.
\newblock In {\em Proceedings of the 49th Annual IEEE Symposium on Foundations
  of Computer Science ({FOCS})}, pages 447--456, 2008.

\bibitem{KempeKobayashiMatsumotoVidick2009}
J.~Kempe, H.~Kobayashi, K.~Matsumoto, and T.~Vidick.
\newblock Using entanglement in quantum multi-prover interactive proofs.
\newblock {\em Computational Complexity}, 18(2):273--307, 2009.

\bibitem{KempeRegevToner2010}
J.~Kempe, O.~Regev, and B.~Toner.
\newblock Unique games with entangled provers are easy.
\newblock {\em SICOMP}, 39(17):3207--3229, 2010.

\bibitem{KobayashiMatsumoto2003}
H.~Kobayashi and K.~Matsumoto.
\newblock Quantum multi-prover interactive proof systems with limited prior
  entanglement.
\newblock {\em J. Comput. Syst. Sci}, 66(3):429--450, 2003.

\bibitem{KochenSpecker1967}
S.~Kochen and E.~Specker.
\newblock The problem of hidden variables in quantum mechanics.
\newblock {\em Journal of Mathematics and Mechanics}, 17:59--87, 1967.

\bibitem{LapidotShamir1991}
D.~Lapidot and A.~Shamir.
\newblock Fully parallelized multi prover protocols for {NEXP}-time.
\newblock In {\em Proceedings of the 32th Annual IEEE Symposium on Foundations
  of Computer Science ({FOCS})}, pages 13--18, 1991.

\bibitem{LeungTonerWatrous2008}
D.~Leung, B.~Toner, and J.~Watrous.
\newblock Coherent state exchange in multi-prover quantum interactive proof
  systems.
\newblock Avaliable at ar{X}iv.org eprint archive, arXiv:0804.4118v1
  [quant-ph].

\bibitem{Mermin1990}
N.~D. Mermin.
\newblock Quantum mysteries revisited.
\newblock {\em American Journal of Physics}, 58:731--734, 1990.

\bibitem{NielsenChuang2000}
M.~Nielsen and I.~Chuang.
\newblock {\em {Q}uantum {C}omputation and {Q}uantum {I}nformation}.
\newblock {C}ambridge {U}niversity {P}ress, Cambridge, 2000.

\bibitem{Peres}
A.~Peres.
\newblock Incompatible results of quantum measurements.
\newblock {\em Physical Review Letters A}, 151:107--108, 1990.

\bibitem{Strauman2008}
N.~Straumann.
\newblock A simple proof of the {K}ochen-{S}pecker theorem on the problem of
  hidden variables.
\newblock {\em Annalen der Physik}, 19:121--127, 2009.

\bibitem{vanDamHayden2002}
W.~van Dam and P.~Hayden.
\newblock Universal entanglement transformations without communication.
\newblock {\em Phys. Rev. A}, 67(6):060302, Jun 2003.

\bibitem{Wehner2006}
S.~Wehner.
\newblock Entanglement in interactive proof systems with binary answers.
\newblock In {\em Proceedings of the 23rd Symposium on Theoretical Aspects of
  Computer Science ({STACS})}, volume 3884 of {\em Lecture Notes in Computer
  Science}, pages 162--171. Springer, 2006.

\end{thebibliography}

\end{document}